\documentclass[11pt]{article}

\usepackage{mathtools}
\usepackage{algorithm}
\usepackage{comment}
\usepackage{algpseudocode}
\usepackage{fullpage}
\usepackage{hyperref}

\input{common_macros.tex}
\input{lazy_macros.tex}
\usepackage{amsthm}
\usepackage{thmtools,thm-restate}
\usepackage[nameinlink,capitalise]{cleveref}

\numberwithin{equation}{section}
\declaretheoremstyle[bodyfont=\it,qed=\qedsymbol]{noproofstyle}

\declaretheorem[name=Observation,numbered=no]{observation*}

\declaretheorem[numberlike=equation]{theorem}

\declaretheorem[name=Theorem,numbered=no]{theorem*}

\declaretheorem[numberlike=equation]{lemma}
\declaretheorem[name=Lemma,numbered=no]{lemma*}

\declaretheorem[name=Corollary,numbered=no]{corollary*}

\declaretheorem[name=Proposition,numbered=no]{proposition*}

\declaretheorem[name=Claim,numbered=no]{claim*}

\declaretheorem[name=Conjecture,numbered=no]{conjecture*}

\declaretheorem[name=Question,numbered=no]{question*}

\declaretheoremstyle[bodyfont=\it,qed=$\lozenge$]{defstyle} 

\declaretheorem[numberlike=equation,style=defstyle]{definition}
\declaretheorem[unnumbered,name=Definition,style=defstyle]{definition*}

\declaretheorem[numberlike=equation,style=defstyle]{example}
\declaretheorem[unnumbered,name=Example,style=defstyle]{example*}

\declaretheorem[unnumbered,name=Notation=defstyle]{notation*}

\declaretheorem[unnumbered,name=Construction,style=defstyle]{construction*}

\declaretheorem[unnumbered,name=Remark,style=defstyle]{remark*}

\title{Pseudo-Deterministic Construction of Irreducible Polynomials over Finite Fields}

\author{
    {Shanthanu S Rai \thanks{Tata Institute of Fundamental Research, Mumbai, India. Email: \texttt{shanthanu.rai@tifr.res.in}.  Research supported by the Department of Atomic Energy, Government of India, under project 12-R{\&}D-TFR-5.01-0500.}}
}
\date{}

\begin{document}

\maketitle

\begin{abstract}
We present a polynomial-time pseudo-deterministic algorithm for constructing
irreducible polynomial of degree $d$ over finite field $\mathbb{F}_q$. A
pseudo-deterministic algorithm is allowed to use randomness, but with high
probability it must output a canonical irreducible polynomial. Our construction
runs in time $\tilde{O}(d^4 \log^4{q})$. 

Our construction extends Shoup's deterministic algorithm (FOCS 1988) for the
same problem, which runs in time $\tilde{O}(d^4 p^{\frac{1}{2}} \log^4{q})$
(where $p$ is the characteristic of the field $\mathbb{F}_q$). Shoup had shown
a reduction from constructing irreducible polynomials to factoring polynomials
over finite fields. We show that by using a fast randomized factoring algorithm,
the above reduction yields an efficient pseudo-deterministic algorithm for
constructing irreducible polynomials over finite fields.
    
\end{abstract}

\section{Introduction}
\label{section:introduction}

A polynomial $f(X)$ over a finite field $\Fq$ ($q$ is a prime power) is
said to be irreducible if it doesn't factor as $f(X) = g(X) h(X)$ for some
non-trivial polynomials $g(X)$ and $h(X)$. Irreducible polynomials over finite
fields are algebraic analogues of primes numbers over integers. It is natural to
ask if one can construct an irreducible polynomial of degree $d$ over $\Fq$
efficiently. Constructing these irreducible polynomials are important since they
yield explicit construction of finite fields of non-prime order. Working over
such non-prime finite fields is crucial in coding theory, cryptography,
pseudo-randomness and derandomization. Any algorithm that constructs irreducible
polynomials of degree $d$ over $\Fq$ would output $d \log{q}$ bits, so we expect
an efficient algorithm for constructing irreducible polynomials would run in
time $\poly(d, \log{q})$.

About $\frac{1}{d}$ fraction of polynomials of degree $d$ are irreducible over
$\Fq$~\cite[Ex. 3.26 and 3.27]{lidl-niederreiter-94}. This gives a simple
``trial and error'' randomized algorithm for constructing irreducible
polynomials, namely, pick a random degree $d$ polynomial and check if it is
irreducible. We can use Rabin's algorithm~\cite{rabin-80} for checking if a
polynomial is irreducible, which can be implemented in
$\tilde{O}(d\log^2{q})$~\cite[Section
8.2]{kedlaya-umans-11}\footnote{$\tilde{O}$ notation omits $\log$ factors in $d$
and $\log{q}$}. In order to improve the probability of finding an irreducible
polynomial to $\frac{1}{2}$, we sample about $d$ polynomials of degree $d$ and
check if any one of them is irreducible. Thus, the ``trial and error'' algorithm
runs in time $\tilde{O}(d^2 \log^2{q})$. Couveignes and
Lercier~\cite{couveignes-lercier-2013} give an alternative randomized algorithm
that runs in time $\tilde{O}(d \log^5{q})$, which is optimal in the exponent of
$d$. Their algorithm constructs irreducible polynomials by using isogenies
between elliptic curves.

Motivated by this, it is natural to ask if there is also an efficient
deterministic algorithm for constructing irreducible polynomials. In the 80s,
some progress was made towards this problem. Adleman and
Lenstra~\cite{adleman-lenstra-86} gave an efficient deterministic algorithm for
this problem conditional on the generalized Riemann hypotheses. They also gave
an unconditional deterministic algorithm which outputs an irreducible polynomial
of degree approximately $d$. Shoup~\cite{shoup-88} gives a deterministic
algorithm of constructing degree $d$ irreducible polynomial which runs in time
$\tilde{O}(d^4 p^{\frac{1}{2}} \log^4{q})$ (where $p$ is the characteristic of
$\Fq$). So, Shoup's algorithm is efficient for fields of small characteristic
($p << d$). But when $p$ is large (say super exponential in $d$), the algorithm
does not run in polynomial time due to the $p^{\frac{1}{2}}$ factor in the run
time. Since then, there hasn't been much progress towards this problem and in
particular, the problem of efficient and unconditional deterministic
construction of irreducible polynomials over $\Fq$ remains open! In fact, the
special case of efficient and unconditional deterministic construction of
quadratic non-residues in $\Fp$ is also open.

One can ask similar questions in the integer world, namely, ``How to efficiently
construct $n$-bit prime numbers?''. By the Prime Number Theorem, there are about
$\frac{1}{n}$ $n$-bit prime numbers (note the similarity between density of
primes and density of irreducible polynomials over $\Fq$). Again this gives a
simple randomized algorithm of just sampling a random $n$-bit number and
checking if it's prime using AKS primality test~\cite{agrawal-kayal-saxena-05}.
But here too, there is no known efficient deterministic algorithm for
constructing $n$-bit prime numbers~\cite{tao-croot-helfgott-12}. 

Due to the difficulty in finding deterministic algorithms for these problems, we
ask a slightly weaker but related question. Are there efficient
pseudo-deterministic algorithms for these problems?

\begin{definition}
\label{def:pseudo-det-algo}
A pseudo-deterministic algorithm is a randomized algorithm which for a given
input, generates a canonical output with probability at least $\frac{1}{2}$.
\end{definition}

Gat and Goldwasser~\cite{gat-goldwasser-2011} first introduced the notion
of pseudo-deterministic algorithm (they had called it Bellagio algorithm).
Pseudo-deterministic algorithm can be viewed as a middle ground between a
randomized and a deterministic algorithm. From an outsider's perspective, a
pseudo-deterministic algorithm seems like a deterministic algorithm in the sense
that with high probability it outputs the same output for a given input.
The breakthrough result of Chen et al.~\cite{chen-lu-oliveira-santhanam-2023}
gave a polynomial-time pseudo-deterministic algorithm for constructing $n$-bit
prime numbers in the infinitely often regime.

\begin{theorem}
    There is a randomized polynomial-time algorithm $B$ such that, for
    infinitely many values of $n$, $B(1^n)$ outputs a canonical $n$-bit prime
    $p_n$ with high probability.
\end{theorem}

In particular, their algorithm doesn't give valid outputs for all values of the
input $n$. Surprisingly, their algorithm is based on complexity theoretic ideas,
and not number theoretic ideas. In fact, they show a more general result that if
a set of strings $Q$ are ``dense'' and it is ``easy'' to check if a string
$x$ is in $Q$, then there is an efficient pseudo-deterministic algorithm for
generating elements of $Q$ of a particular length in the infinitely often
regime. Both prime numbers and irreducible polynomials over $\Fq$ satisfy this
property. Thus, this gives an efficient pseudo-deterministic algorithm for
constructing irreducible polynomials over $\Fq$ in the infinitely often regime.

But not only does this algorithm not work for all $d$, there are no good density
bounds for the fraction of $d$ where the algorithm gives valid output. So it is
natural to ask if we can extend this result to all values of degree $d$ over all
finite fields $\Fq$. In this paper, we present a more direct
pseudo-deterministic algorithm for constructing irreducible polynomials over
$\Fq$ (for all degrees $d$) which crucially relies on the structure of
irreducible polynomials. Our result extends Shoup's~\cite{shoup-88}
deterministic algorithm for constructing irreducible polynomials. Shoup reduces
the problem of constructing irreducible polynomials to factoring polynomials
over $\Fq$. We observe that by making use of the fast randomized factoring
algorithm, and the ``canonization'' process described by Gat and
Goldwasser~\cite{gat-goldwasser-2011} for computing $q$-th residues over $\Fp$,
the above reduction yields an efficient pseudo-deterministic algorithm for
constructing irreducible polynomials over $\Fq$.

\begin{theorem}
    \label{theorem:main-theorem}
    There is a pseudo-deterministic algorithm for constructing an irreducible
    polynomial of degree $d$ over $\Fq$ ($q$ is prime power) in expected time
    $\tilde{O}(d^4 \log^4{q})$.
\end{theorem}

\section{Overview}

As mentioned earlier, Shoup's deterministic algorithm~\cite{shoup-88} is
efficient for fields of small characteristic. We extend Shoup's algorithm and
make it efficient over all fields, but at the cost of making the algorithm
pseudo-deterministic. In order to see the main ideas involved, let's consider a
toy problem of constructing irreducible polynomial of degree 2 over $\Fp$ ($p$
is prime). Suppose we could get our hands on some quadratic non residue
$\alpha$, then $X^2 - \alpha$ would be irreducible. There are $\frac{p-1}{2}$
quadratic non residues in $\Fp$, so if we randomly pick an $\alpha \in \Fp$ and
output $X^2 - \alpha$, it would be irreducible with about $\frac{1}{2}$
probability. But this approach wouldn't be pseudo-deterministic, since in each
run we will very likely choose different $\alpha$.

In order to obtain a canonical quadratic non residue $\alpha$, we first set
$\alpha = -1$ and repeatedly perform $\alpha \gets \sqrt{\alpha}$ (choosing the
smallest square root) until $\alpha$ is a quadratic non residue. Here, $\beta$
is a square root of $\alpha$ if $\beta^2 = \alpha \pmod{p}$. For computing the
square root, we can use Cantor-Zassenhaus randomized factoring
algorithm~\cite{cantor-zassenhaus-81}. In \cref{example:qnr}, we illustrate the
above strategy over a specific finite field.
\cref{algo:pseudo-det-degree-2-irred}
implements this strategy.

\begin{example}
    \label{example:qnr}
    Let's try to pseudo-deterministically construct a quadratic non-residue in
    $\F_{73}$. We first set $\alpha = -1$.
    Square roots of $-1 \pmod{73}$ are $27$ and $46$. The square roots are
    computed using Cantor-Zassenhaus randomized factoring
    algorithm~\cite{cantor-zassenhaus-81}. 
    
    We choose the smallest square root $27$ and set $\alpha = 27$. Square roots
    of $27 \pmod{73}$ are $10$ and $63$. We choose the smallest square root $10$
    and set $\alpha = 10$. Since $10$ is a quadratic non-residue, we output $10$
    (we use Euler's criterion\footnote{Euler's criterion: For odd prime $p$, $a$
    is a quadratic non residue iff $a^{(p-1)/2}=-1$} to check if $10$ is a
    quadratic non-residue).
\end{example}

\begin{algorithm}[H]
    \caption{Pseudo-deterministically constructing irreducible polynomial of degree 2 over $\Fp$}
    \label{algo:pseudo-det-degree-2-irred}
    % \hspace*{\algorithmicindent} \textbf{Input:} Prime $p = 2^k l + 1$, where $l$ is odd
    % \hspace*{\algorithmicindent} \textbf{Output:} Irreducible polynomial of degree 2 over $\F_p$
    \begin{algorithmic}[1]
        \State $\alpha \gets -1$
        \While{$\alpha$ is a quadratic residue} \label{line:while-condition}
            \State Factorize $X^2 - \alpha = (X - \beta_1) (X - \beta_2)$ \label{line:factorize}
            \State $\alpha \gets \min(\beta_1, \beta_2)$
        \EndWhile
        \State Output $X^2 - \alpha$
    \end{algorithmic}
\end{algorithm}

Suppose $p - 1 = 2^k l$ (where $l$ is odd). Each time we take square root, the
order\footnote{Order of $\alpha$ is the least integer $k > 0$ such that
$\alpha^k = 1$ in $\Fp$} of $\alpha \pmod p$ doubles. Since the order of
$\alpha$ divides $\abs{\Fp^{*}} = 2^k l$ (by Lagrange's theorem), we can
repeatedly take square roots in \cref{algo:pseudo-det-degree-2-irred} at most
$k$ times. Thus, \cref{algo:pseudo-det-degree-2-irred} will terminate with at
most $\log{p}$ iterations of the while loop. This algorithm is based on Gat and
Goldwasser's algorithm~\cite{gat-goldwasser-2011} for computing $q$-th residues
over $\Fp$. The algorithm is pseudo-deterministic since at each iteration of the
while loop, we ``canonize'' our choice of square root by picking the smallest
one among the two choices. Note that we used Euler's criterion for checking if
$\alpha$ is a quadratic residue or not in Line~\ref{line:while-condition}.

We can generalize the above ideas for constructing irreducible polynomials over
finite fields. Shoup~\cite{shoup-88} showed that constructing irreducible
polynomials over $\Fp$ reduces to finding $q$-th non residues over appropriate
field extensions ($q$ is prime). These $q$-th non residues can be
pseudo-deterministically constructed using similar techniques as in
\cref{algo:pseudo-det-degree-2-irred}.

The rest of the paper is organized as follows. We start with some preliminaries
in \cref{section:preliminaries}. In
\cref{section:construction-over-extension-field}, we will reduce the problem of
constructing irreducible polynomials over extensions fields $\F_{p^k}$ to
constructing them over $\Fp$. \cref{section:construction-over-prime-field} will
make use of Shoup's observation mentioned in previous paragraph to construct
irreducible polynomials over $\Fp$. Finally, in \cref{section:conclusion} we
conclude with some open problems.

\section{Preliminaries}
\label{section:preliminaries}

\subsection{Pseudo-deterministic algorithms}

We defined pseudo-deterministic algorithm to be randomized algorithm which for a
given input, generates a canonical output with probability at least
$\frac{1}{2}$. In this paper, whenever the pseudo-deterministic algorithm
doesn't generate a canonical output, it just fails and doesn't give any valid
output. In such cases, we can just rerun the algorithm until we get some valid
output (which is bound to be canonical). Now the runtime of the algorithm will
be random, but the expected runtime will be (asymptotically) same as the
original runtime.

For all the pseudo-deterministic algorithms in this paper, we report the
expected run time in the above sense. These algorithms always generate a
canonical output, but the amount of time they take to do so is random.

\subsection{Finite Field primer}

In this subsection, we go over some basic facts about finite fields that will
be useful in later sections.

\subsubsection{Splitting field}

A polynomial $h(X) \in \K[X]$ may not factorize fully into linear factors over
the field $\K$. Suppose $\F$ is the smallest extension of $\K$ such  that $h(X)$
fully factorizes into linear factors over $\F$. In other words, there exists
$\alpha_1, \alpha_2, \dots \alpha_k \in \F$ such that,

\[
    h(X) = (X - \alpha_1) (X - \alpha_2) \cdots (X - \alpha_k)
\]

Then $\F$ is the called the splitting field of $h(X)$ over $\K$~\cite[Definition
1.90]{lidl-niederreiter-94}. Note that for any other extension of $\K$ that is a
proper subfield of $\F$, $h(X)$ will not fully factorize into linear factors.

\subsubsection{Structure of Finite Fields}

For every prime power $p^n$ ($p$ is prime), there exists a finite field of size
$p^n$ and all finite fields of size $p^n$ are isomorphic to each
other~\cite[Theorem 2.5]{lidl-niederreiter-94}.

\begin{theorem}[Existence and Uniqueness of Finite Fields]
    For every prime $p$ and every positive integer $n$ there exists a finite
    field with $p^n$ elements. Any finite field with $q=p^n$ elements is
    isomorphic to the splitting field of $X^q - X$ over $\Fp$.
\end{theorem}

Thus, elements of $\F_{p^n}$ are roots of $X^{p^n} - X$. From this, we get the
following generalization of Fermat's little theorem for finite fields:

\begin{theorem}[Fermat's little theorm for finite fields]
    \label{theorem:fermat-finite-field}
    If $\alpha \in \F_{p^n}$, then $\alpha^{p^n} = \alpha$. Conversely, if
    $\alpha$ is in some finite field and $\alpha^{p^n} = \alpha$, then
    $\alpha \in \F_{p^n}$.
\end{theorem}

The below theorem gives the necessary and sufficient condition for a finite
field $\F_{p^m}$ to be a subfield of another finite field
$\F_{p^n}$~\cite[Theorem 2.6]{lidl-niederreiter-94}.

\begin{theorem}[Subfield Criterion]
    \label{theorem:subfield-criterion}
    Let $\Fq$ be a finite field with $q=p^n$ elements. Then every subfield of
    $\Fq$ has order $p^m$, where $m$ is the positive divisor of $n$. Conversely,
    if $m$ is the positive divisor of $n$, then there is exactly one subfield of
    $\Fq$ with $p^m$ elements.
\end{theorem}

From \cref{theorem:fermat-finite-field} and \cref{theorem:subfield-criterion},
we get the following useful lemma:

\begin{lemma}
    \label{lemma:smallest-k-subfield}
    Suppose $\alpha$ is some finite field element. Let $k$ be the
    smallest integer greater than 0 such that $\alpha^{p^k} = \alpha$. Then,
    $\F_{p^k}$ is the smallest extension of $\Fp$ that contains $\alpha$. In
    other words, $\alpha \in \F_{p^k}$ and for all $1 \leq k' < k$, $\alpha
    \notin \F_{p^{k'}}$.  
\end{lemma}

% For representing elements in extension field $\F_{p^n}$, we first get our hands
% on an irreducible polynomial of degree $n$ over $\Fp$. Then, $\Fp[X]/(f(X))$ is
% isomorphic to the extension field $\F_{p^n}$~\cite[Chapter 2 Section
% 5]{lidl-niederreiter-94}.

\subsubsection{Conjugates and Minimal polynomial}
\label{section:minimal-poly-conjugates}

Let $f(X)$ be an irreducible polynomial of degree $n$ over $\Fq$ ($q$ is prime
power). Then, $f(X)$ has some root $\alpha \in \F_{q^n}$. Also, the elements
$\alpha, \alpha^q, \alpha^{q^2}, \dots, \alpha^{q^{n-1}}$ are all distinct and
are the roots of $f(X)$~\cite[Theorem 2.14]{lidl-niederreiter-94}.

\[
    f(X) = (X - \alpha) (X - \alpha^q) (X - \alpha^{q^2}) \cdots (X - \alpha^{q^{n-1}})
\]

The splitting field of $f(X)$ with respect to $\Fq$ is
$\F_{q^n}$~\cite[Corollary 2.15]{lidl-niederreiter-94}. The minimal polynomial
of $\alpha$ over $\Fq$ is $f(X)$.

Above, the roots of $f(X)$ are all of the form $\alpha^{q^i}$. We will call such
elements conjugates $\alpha$ with respect to $\Fq$:

\begin{definition}
    Let $\F_{q^n}$ be an extension of $\Fq$ and let $\beta \in \F_{q^n}$.
    Then, $\beta, \beta^q, \beta^{q^2}, \dots, \beta^{q^{n-1}}$
    are called the conjugates of $\beta$ with respect to $\Fq$.
\end{definition}

In the following sections, we will be using the below lemma to show that certain
polynomials are irreducible.

\begin{lemma}[Minimal polynomial of $\beta \in \F_{q^n}$]
    \label{lemma:minimal-poly-of-conjugate}
    Suppose $\beta \in \F_{q^n}$ and conjugates of $\beta$ with respect to $\Fq$
    are all distinct. Then the minimal polynomial of $\beta$ over $\Fq$
    has degree $n$ and is of the form:

    \[
        g(X) = (X - \beta) (X - \beta^q) (X - \beta^{q^2}) \cdots (X - \beta^{q^{n-1}})
    \]

    Thus, $g(X) \in \Fq[X]$ is an irreducible polynomial.
\end{lemma}
\begin{proof}
    The minimal polynomial $g(X)$ of $\beta$ over $\Fq$ is the smallest degree
    polynomial in $\Fq[X]$ such that $g(\beta) = 0$. Since, $g(\beta^{q^i}) =
    g(\beta)^{q^i} = 0$, all conjugates of $\beta$ are roots of $g(X)$. Hence,
    degree of $g(X)$ is at least $n$ (since the conjugates are all distinct).
    Also since $\beta \in \F_{q^n}$, degree of $g(X)$ is at most $n$. Thus, the
    degree of $g(X)$ is $n$.

    Thus, $g(X) = (X - \beta) (X - \beta^q) (X - \beta^{q^2}) \cdots (X -
    \beta^{q^{n-1}})$. Since $g(X)$ is a minimal polynomial of $\beta$ over
    $\Fq$, it will be in $\Fq[X]$ and is irreducible.
\end{proof}

\subsubsection{Representing finite field elements}
\label{section:represent-finite-field-elements}

Throughout the paper, we assume that extension fields $\F_{p^k}$ are given to us
as $\Fp[X]/(f(X))$, where $f(X)$ is an irreducible polynomial of degree $k$ over
$\Fp$ (refer~\cite{lenstra-91} for working with other representations). Each
element in $\Fp[X]/(f(X))$ can be viewed as a polynomial with degree at most $k$
over $\Fp$. The coefficient vectors of these polynomials are in $\Fp^k$. This
gives a natural isomorphism $\Phi: \F_{p^k} \rightarrow \Fp^k$. In $\Fp^k$, we
can order elements in lexicographic order in the natural sense.

\begin{definition}
    We say that $\alpha \in \F_{p^k}$ is lexicographically smaller than $\beta
    \in \F_{p^k}$, if $\Phi(\alpha) \in \Fp^k$ is lexicographically smaller than
    $\Phi(\beta) \in \Fp^k$.
\end{definition}

In the above definition, we compare the coordinates of $\Phi(\alpha)$ and
$\Phi(\beta)$ by fixing some ordering on elements on $\Fp$ (for e.g., we can
consider the natural ordering one gets from the additive group structure of
$\Fp$). Checking if $\Phi(\alpha)$ is lexicographically smaller than
$\Phi(\beta)$ requires $k$ comparisons, with each comparison taking $O(\log{p})$
time. Thus, overall it takes $O(k\log{p})$ time to check if $\Phi(\alpha)$ is
lexicographically smaller than $\Phi(\beta)$. 

Similarly, we can define a lexicographic ordering on polynomials over
$\F_{p^k}$.

\begin{definition}
    Suppose we are given two polynomials $g(X)$ and $h(X)$ of degree $d$ over
    $\F_{p^k}$. Then we say that $g(X)$ is lexicographically smaller than $h(X)$
    if the coefficient vector of $g(X)$ is lexicographically smaller than
    coefficient vector of $h(X)$ (the coefficients are compared using $\Phi$).    
\end{definition}

Checking if $g(X)$ is lexicographically smaller than $h(X)$ requires $d+1$
comparisons, with each comparison taking $O(k\log{p})$ time. Thus, overall it
takes $O(dk\log{p})$ time to check if $g(X)$ is lexicographically smaller than
$h(X)$.

\begin{lemma}[Picking lexicographically smallest polynomial]
    \label{lemma:lex-smallest-poly-find}
    Suppose we are given $n$ polynomials $f_1(X), f_2(X), \dots, f_n(X)$ of
    degree $d$ over $\F_{p^k}$. Then, there is an algorithm that outputs the
    lexicographically smallest polynomial among them in $O(ndk\log{p})$ time.
\end{lemma}
\begin{proof}
    We go over each polynomial $f_i(X)$ one by one, checking if $f_i(X)$ is
    lexicographically smaller than the lexicographically smallest polynomial we
    have seen so far. Since each comparison takes $O(dk\log{p})$ time, and we do
    at most $n$ comparisons, the algorithm runs in $O(ndk\log{p})$ time.
\end{proof}

% \subsubsection{$q$-th non-residues}

% finite fields $q$-th residue
% might not exist always
% Euler criterion for checking if qth non residue

\subsection{Equal degree polynomial factorization}

Shoup~\cite{shoup-88} reduced constructing irreducible polynomials to factoring
polynomials over finite field. It turns out that the reduction factors
polynomials whose irreducible factors all have same degree. Hence, equal degree
factorization is a crucial sub-routine for constructing irreducible polynomials.
There are several fast randomized equal degree factorization algorithms, and
below we mention one of them:

\begin{theorem}[Equal degree factorization]
    \label{theorem:equal-degree-factorization}
    Suppose $f(X)$ is a polynomial of degree $d$ over $\Fq$ ($q$ is prime power)
    which factors into irreducible polynomials of equal degree. Then, the equal
    degree factorization algorithm by von zur Gathen \&
    Shoup~\cite{von-zur-gathen-shoup-92} factors $f(X)$ in expected time
    $\tilde{O}(d\log^2{q})$.
\end{theorem}

\section{Construction of irreducible polynomials over extension fields $\F_{p^k}$}
\label{section:construction-over-extension-field}

We first show in \cref{algo:construct-irred-over-extension-field} that
constructing irreducible polynomials over extension fields $\F_{p^k}$ can be
reduced to constructing irreducible polynomials over $\Fp$ ($p$ is prime).
\cref{theorem:construct-irred-over-extension-field} shows the correctness and
running time of \cref{algo:construct-irred-over-extension-field}.

\begin{algorithm}[H]
    \caption{Pseudo-deterministic construction of irreducible polynomials over $\F_{p^k}$}
    \label{algo:construct-irred-over-extension-field}
    \hspace*{\algorithmicindent} \textbf{Input:} Degree $d$     \\
    \hspace*{\algorithmicindent} \textbf{Output:} Irreducible polynomial of degree $d$ over $\F_{p^k}$
    \begin{algorithmic}[1]
        \State Pseudo-deterministically construct irreducible polynomial $f(X)$ over $\Fp$ of degree $dk$.

        \State Factor $f(X) = \prod_{i=0}^{k-1} f_i(X)$ over $\F_{p^k}$ using \cref{theorem:equal-degree-factorization}.
        
        \State Output the lexicographically smallest factor $f_i(X)$.
    \end{algorithmic}
\end{algorithm}

\begin{theorem}[Correctness and Running time of \cref{algo:construct-irred-over-extension-field}]
    \label{theorem:construct-irred-over-extension-field}
    Suppose there is a pseudo-deterministic algorithm for constructing
    irreducible polynomials of degree $l$ over $\Fp$ ($p$ prime), that runs in
    expected time $T(l, p)$. Then \cref{algo:construct-irred-over-extension-field}
    pseudo-deterministically constructs irreducible polynomials of degree $d$
    over extension field $\F_{p^{k}}$ in expected time $T(dk, p) +
    \tilde{O}(dk^3\log{p})$.
\end{theorem}
\begin{proof}
    \cref{algo:construct-irred-over-extension-field} first constructs an
    irreducible polynomial $f(X)$ of degree $dk$ over $F_p$. Note that
    $\Fp[X]/(f(X))$ is isomorphic to $\F_{p^{dk}}$. Some $\alpha \in
    \F_{p^{dk}}$ will be a root of $f(X)$. The conjugates of $\alpha$ with
    respect to $\Fp$ are all distinct and are the roots of $f(X)$ (refer
    \cref{section:minimal-poly-conjugates}):

    \begin{align*}
        f(X) = (X - \alpha) (X - \alpha^{p}) (X - \alpha^{p^2}) \cdots (X - \alpha^{p^{dk-2}}) (X - \alpha^{p^{dk-1}})
    \end{align*}

    Rearranging the above terms, we get:

    \begin{align*}
        f(X)=   &   \Big[ (X - \alpha) (X - \alpha^{p^k}) (X - \alpha^{p^{2k}}) \cdots  (X-\alpha^{p^{(d-1)k}}) \Big]                \\
                &   \Big[ (X - \alpha^{p}) (X - \alpha^{p^{k+1}}) (X - \alpha^{p^{2k+1}}) \cdots  (X-\alpha^{p^{(d-1)k+1}}) \Big]    \\
                &   \Big[ (X - \alpha^{p^2}) (X - \alpha^{p^{k+2}}) (X - \alpha^{p^{2k+2}}) \cdots  (X-\alpha^{p^{(d-1)k+2}}) \Big]  \\
                &   \vdots   \\
                &   \Big[ (X - \alpha^{p^{(k-1)}}) (X - \alpha^{p^{k+(k-1)}}) (X - \alpha^{p^{2k+(k-1)}}) \cdots  (X-\alpha^{p^{(d-1)k+(k-1)}}) \Big]  \\
            =   &   \prod_{i=0}^{k-1} \prod_{j=0}^{d-1} (X - \alpha^{p^{jk + i}})   \\
            \coloneq  &   \prod_{i=0}^{k-1} f_i(X)
    \end{align*}

    Let $q = p^k$. $f_i(X) $ has degree $d$ and its roots are conjugates of
    $\alpha^{p^i} \in \F_{q^d}$ with respect to $\F_q$ (which are all distinct).
    Thus, from \cref{lemma:minimal-poly-of-conjugate}, $f_i(X) \in \F_q[X]$ is
    the minimal polynomial of $\alpha^{p^i}$ over $\Fq$, and hence $f_i(X)$ is
    irreducible over $\F_q$. So, we can use
    \cref{theorem:equal-degree-factorization} to factorize $f(X)$ over $\F_q$,
    obtaining all factors $f_i(X)$ of degree $d$. We then use
    \cref{lemma:lex-smallest-poly-find} to output the lexicographically smallest
    factor among $f_i(X)$. Let the lexicographically smallest factor be denoted
    by $f_{i^*}(X)$. Given a polynomial $f(X)$ of degree $dk$, $f_{i^*}(X)$ is
    canonical. Thus, the above construction is pseudo-deterministic.

    For the running time, it takes $T(dk, p)$ time to construct $f(X)$, and then
    $\tilde{O}(dk^3\log{p})$ time to factor $f(X)$ over field $\F_{p^k}$ (from
    \cref{theorem:equal-degree-factorization}). Finally, choosing $f_{i^*}(X)$
    among $f_i(X)$ can be computed in time ${O}(dk^2\log{p})$ (from
    \cref{lemma:lex-smallest-poly-find}). Thus, the overall running time of the
    algorithm is $T(dk, p) + \tilde{O}(dk^3\log{p})$.
\end{proof}

\section{Construction of irreducible polynomials over $\Fp$}
\label{section:construction-over-prime-field}

Shoup's algorithm reduces constructing irreducible polynomials over $\Fp$ to
finding $q$-th non residues in splitting field of $X^q - 1$, for all prime
divisors $q$ of $d$ (and $q \neq p$). For completeness, we reproduce the theorem
below and refer to Theorem 2.1 in~\cite{shoup-88} for it's proof.

\begin{theorem}[Reduction to finding $q$-th non residues] 
    \label{theorem:reduction-to-qth-non-residue}
    
    Assume that for each prime $q \mid d$, $q \neq p$, we are given a splitting
    field $\K$ of $X^q - 1$ over $\Fp$ and a $q$-th non residue in
    $\K$. Then we can find an irreducible polynomial over $\Fp$ of
    degree $d$ deterministically with $\tilde{O}(d^4 \log{p} + \log^2{p})$
    operations in $\Fp$.
\end{theorem}

Shoup constructs the splitting field\footnote{Shoup constructs an irreducible
polynomial $g(X) \in \Fp[X]$ such that $\Fp/(g(X))$ is isomorphic to splitting
field of $X^q - 1$ over $\Fp$} $\K$ of $X^q - 1$ over $\Fp$ and a $q$-th non
residue in $\K$ by reducing to deterministic polynomial factorization. Since no
known efficient deterministic factoring algorithms are known, his algorithm is
not efficient for finite fields of large characteristic. In this section, we
will find a canonical splitting field $\K$ and a canonical $q$-th non residues
by using a fast randomized factoring algorithm. Thus, we obtain an efficient
algorithm for constructing irreducible polynomials over $\Fp$, but at the cost
of making the algorithm pseudo-deterministic.

% For each prime $q \mid d, q \neq p$, we will pseudo-deterministically construct
% a splitting field $\K$ of $X^q - 1$ over $\Fp$ and find a $q$-th non
% residue in $\K$. Then using \cref{theorem:reduction-to-qth-non-residue},
% we can construct an irreducible polynomial of degree $d$ over $\Fp$.

For each prime $q \mid d, q \neq p$, we will pseudo-deterministically construct
a splitting field $\K$ of $X^q - 1$ over $\Fp$ and find a $q$-th non
residue in $\K$. To this end, we first analyze the factorization of $X^q
- 1 \in \Fp[X]$.

\begin{lemma}
    \label{lemma:factor-xq-1}
    Consider the polynomial $X^q - 1 \in \Fp[X]$ ($p, q$ are prime numbers). Let
    $k$ be the smallest integer greater than 0 such that $q \mid p^k - 1$ (in
    other words, $k$ is the order of $p \pmod q$). Then,

    \begin{enumerate}
        \item The splitting field of $X^q - 1$ over $\Fp$ is $\F_{p^k}$
        \item $X^q - 1 = (X - 1) g_1(X) g_2(X) \cdots g_{\frac{q-1}{k}}(X)$
        
                where $g_i(X) \in \Fp[X]$ are irreducible polynomials of degree $k$.
    \end{enumerate}
\end{lemma}
\begin{proof}
    
    Let $\K$ be the splitting field of $X^q - 1$ over $\Fp$. The roots of $X^q -
    1$ in $\K$ are by definition the $q$-th roots of unity. Suppose $\omega \in
    \K$ is some primitive $q$-th root of unity. Then, $\{ 1, \omega, \omega^2,
    \ldots, \omega^{q-1} \}$ are all the $q$-th roots of unity, and they form a
    multiplicative subgroup in $\K^*$. In fact, since $q$ is prime, each of $\{
    \omega, \omega^2, \ldots, \omega^{q-1}\}$ is a primitive $q$-th root of
    unity.

    Since $\omega^q = 1$, we have $\omega^{p^k} = \omega$ and hence from
    \cref{theorem:fermat-finite-field}, $\omega \in \F_{p^k}$. By definition of
    $k$, $k$ is the smallest integer greater than 0 such that $\omega^{p^k} =
    \omega$. So from \cref{lemma:smallest-k-subfield}, $\F_{p^k}$ is the
    smallest extension of $\Fp$ that contains $\omega$. $X^q-1$ splits linearly
    as:

    \[
        X^q - 1 = (X - 1) (X - \omega) (X - \omega^2) \cdots (X - \omega^{q-1})
    \]

    $\F_{p^k}$ is the smallest extension of $\Fp$ that contains all the
    roots of $X^q - 1$. Thus, $\F_{p^k}$ is the splitting field of $X^q - 1$.
    We next consider the factorization pattern of $X^q - 1$ over $\Fp$.

    Let $G$ be the multiplicative group of integers modulo $q$. Since $q$ is
    prime, elements of $G$ are $\{1, 2, \dots, q-1\}$. Consider the cyclic
    subgroup $H$ of $G$ generated by $p$. The elements of $H$ are $\{1, p, p^2,
    \dots, p^{k-1}\}$. The cosets of $H$ partition $G$. Let $a_1 H$, $a_2 H$,
    \dots, $a_{(q-1)/k} H$ be the $(q-1)/k$ cosets of $H$ that partition $G$.
    Then, $X^q - 1 \in \Fp[X]$ can be factorized as follows:

    \begin{align*}
        X^q - 1 &= (X - 1) (X - \omega) (X - \omega^2) \cdots (X - \omega^{q-1}) \\
                &= (X - 1) \prod_{i=1}^{(q-1)/k} \prod_{j \in a_i H} (X - \omega^j) \\
                &= (X - 1) \prod_{i=1}^{(q-1)/k} (X - \omega^{a_i}) (X - \omega^{a_i p}) (X - \omega^{a_i p^2}) \cdots (X - \omega^{a_i p^{k-1}}) \\
                &\coloneq (X - 1) \prod_{i=1}^{(q-1)/k} g_i(X)
    \end{align*}

    $g_i(X)$ has degree $k$ and its roots are conjugates of $\omega^{a_i} \in
    \F_{p^k}$ with respect to $\Fp$ (which are all distinct). Thus, from
    \cref{lemma:minimal-poly-of-conjugate}, $g_i(X) \in \Fp[X]$ is the minimal
    polynomial of $\omega^{a_i}$ over $\Fp$, and hence is irreducible over
    $\Fp$.

\end{proof}

Thus, the splitting field of $X^q - 1$ over $\Fp$ is $\F_{p^k}$. It is easy to
see that $\F_{p^k}$ contains a $q$-th non residue, since the map $\Phi: \alpha
\mapsto \alpha^q$ is not surjective in $\F_{p^k}$ (since for every $i$,
$\Phi(\omega^i) = 1$, where $\omega$ is some primitive $q$-th root of unity).

In order to get our hands on a canonical representation of $\F_{p^k}$, we can
factorize $(X^q - 1)/(X - 1) = X^{q-1} + X^{q-2} + \cdots + X + 1$ over $\Fp$
and pick the lexicographically smallest degree $k$ irreducible factor $h(X)$.
Then, $\Fp[X]/h(X)$ is isomorphic to $\F_{p^k}$. Let $\omega$ be an element in
$\F_{p^k}$ isomorphic to $X \in \Fp[X]/h(X)$. Next, to find a canonical $q$-th
non residue $\alpha \in \F_{p^k}$, we set $\alpha = \omega$ and repeatedly
perform $\alpha \gets \sqrt[q]{\alpha}$ (choosing the lexicographically smallest
$q$-th root) until $\alpha$ is a $q$-th non residue.
\cref{algo:construct-irred-mod-p} implements the above idea and constructs an
irreducible polynomial of degree $d$. In Line~\ref{line:factor-xq-1} and
Line~\ref{line:factor-xq-a}, factorization is done using
\cref{theorem:equal-degree-factorization}. We analyze the correctness and
running time of \cref{algo:construct-irred-mod-p} in
\cref{theorem:construct-irred-mod-p}.

\begin{algorithm}
    \caption{Pseudo-deterministic construction of irreducible polynomials over $\Fp$}
    \label{algo:construct-irred-mod-p}
    \hspace*{\algorithmicindent} \textbf{Input:} Degree $d$     \\
    \hspace*{\algorithmicindent} \textbf{Output:} Irreducible polynomial of degree $d$ over $\Fp$
    \begin{algorithmic}[1]
        \State Initialize arrays $H \gets [\ ], \Lambda \gets [\ ]$

        \For{prime $q \mid d, q \neq p$} \label{line:for-each-q}
            \State Factorize $X^{q-1} + X^{q-2} + \cdots + X + 1 = g_1(X) g_2(X) \cdots g_{\frac{q-1}{k}}(X)$ over $\Fp$ \label{line:factor-xq-1}
            \State $h(X) \gets $ lexicographically smallest degree $k$ factor among $g_1(X), g_2(X), \ldots, g_{\frac{q-1}{k}}(X)$ \label{line:choose-h}
            \State Field arithmetic over $\F_{p^k}$ will henceforth be performed over $\Fp[X]/h(X)$.
            \State $\alpha \gets$ element in $\F_{p^k}$ isomorphic to $X$ in $\Fp[X]/h(X)$
            \While{$\alpha$ is a $q$-th residue} \label{line:while-q-residue}
                \State Factorize $X^q - \alpha = (X - \beta_1) (X - \beta_2) \cdots (X - \beta_q)$ over $\F_{p^k}$ \label{line:factor-xq-a}
                \State $\alpha \gets$ lexicographically smallest element among $\beta_1, \beta_2, \ldots, \beta_q$ in $\F_{p^k}$
            \EndWhile
            \State Append $h(X)$ to array $H$ and $\alpha$ to array $\Lambda$
        \EndFor

        \State Using arrays $H$ and $\Lambda$ and
        \cref{theorem:reduction-to-qth-non-residue}, deterministically construct
        an irreducible polynomial of degree $d$ over $\Fp$. \label{line:reduction-to-shoup}
    \end{algorithmic}
\end{algorithm}

\begin{theorem}[Correctness and Runtime of \cref{algo:construct-irred-mod-p}]
    \label{theorem:construct-irred-mod-p}

    % TODO: try to see if the overfull hbox below can be fixed.. not a big deal though (very minor issue)
    \cref{algo:construct-irred-mod-p} pseudo-deterministically constructs an
    irreducible polynomial of degree $d$ over $\Fp$ and runs in expected time
    $\tilde{O}(d^4 \log^3{p})$.
\end{theorem}
\begin{proof}
    We need to show that the for loop in \cref{algo:construct-irred-mod-p}
    correctly computes the splitting field of $X^q - 1$ and finds a $q$-th non
    residue in the splitting field. Then, Line~\ref{line:reduction-to-shoup}
    will correctly output an irreducible polynomial of degree $d$ over $\Fp$
    (from \cref{theorem:reduction-to-qth-non-residue}).

    Let $k$ be the smallest integer greater than 0 such that $q \mid p^k - 1$
    ($k$ is the order of $p \pmod q$). From \cref{lemma:factor-xq-1}, $X^{q-1} +
    X^{q-2} + \cdots + X + 1$ factorizes as $g_1(X) g_2(X) \cdots
    g_{\frac{q-1}{k}}(X)$ where $g_i(X)$ are degree $k$ irreducible polynomials.
    Thus, by choosing the lexicographically smallest degree $k$ irreducible
    factor $h(X)$ of $X^q - 1$, we ensure that the choice of $h(X)$ is
    canonical. $\F_{p^k} \cong \Fp[X]/h(X)$ is the splitting field of $X^q - 1$
    which contains a $q$-th non residue. 
    
    Let $\omega \in \F_{p^k}$ be some primitive $q$-th root of unity. Suppose
    $\alpha$ is a $q$-th residue, and let $\beta \in \F_{p^k}$ such that $\alpha
    = \beta^q$ ($\beta$ is a $q$-th root of $\alpha$). Then, $\{ \beta, \beta
    \omega, \beta \omega^2, \ldots, \beta \omega^{q-1} \}$ are all $q$-th roots
    of $\alpha$. Thus, as required in Line~\ref{line:factor-xq-a}, $X^q -
    \alpha$ will factorize into linear factors. By ensuring that we pick the
    lexicographically smallest $q$-th root of $\alpha$, we ``canonize'' the
    computation of $q$-th non residue. This ``canonization'' process is akin to
    the one Gat and Goldwasser~\cite[Section 5]{gat-goldwasser-2011} used to
    compute $q$-th non residue in $\Fp$.

    But we still need to ensure that the while loop eventually terminates. Let
    $p^k - 1 = q^{\ell} r$, where $r$ is not divisible by $q$. Note that $\ell
    \leq k \log{p}$. In each iteration of the while loop, the order of $\alpha$
    in $\F_{p^k}^*$ increases by a factor of $q$. Since the order of $\alpha$
    divides $\abs{\F_{p^k}^*} = q^{\ell} r$ (by Lagrange's theorem), the while
    loop will terminate in at most $\ell$ steps. Thus, for each prime $q \mid d,
    q \neq p$, the for loop at Line~\ref{line:for-each-q}
    pseudo-deterministically constructs the splitting field $\F_{p^k}$ of $X^q -
    1$ and a $q$-th non residue in $\F_{p^k}$.

    Now we analyze the runtime. From \cref{theorem:equal-degree-factorization},
    equal degree factorization in Line~\ref{line:factor-xq-1} takes $\tilde{O}(q
    \log^2{p})$. From \cref{lemma:lex-smallest-poly-find}, lexicographically
    smallest $h(X)$ in Line~\ref{line:choose-h} can be chosen in $O(q \log{p})$.
    The while loop at Line~\ref{line:while-q-residue} runs at most $\ell$ times.
    The factoring step at Line~\ref{line:factor-xq-a} takes $\tilde{O}(q k^2
    \log^2{p})$ (using \cref{theorem:equal-degree-factorization}) and the
    lexicographically smallest $q$-th root can be picked in $O(qk \log{p})$
    time. Thus, the while loop takes $\tilde{O}(\ell q k^2 \log^2{p})$. Since
    $\ell \leq k\log{p}$ and $k < q$, the running time of the while loop can be
    upper bounded by $\tilde{O}(q^4 \log^3{p})$. Thus the overall running time
    of each iteration of the for loop is $\tilde{O}(q^4 \log^3{p})$. So we can
    upper bound the running time of the entire for loop by $\tilde{O}(d^4
    \log^3{p})$. Since the running time of Line~\ref{line:reduction-to-shoup} is
    also upper bounded by $\tilde{O}(d^4 \log^3{p})$ (from
    \cref{theorem:reduction-to-qth-non-residue}), the overall running time of
    the algorithm is $\tilde{O}(d^4 \log^3{p})$.
\end{proof}

We end this section by completing the proof of \cref{theorem:main-theorem}.

\begin{proof}[Proof of \cref{theorem:main-theorem}]
    \cref{algo:construct-irred-over-extension-field} pseudo-deterministically
    constructs irreducible polynomial of degree $d$ over $\F_{p^k}$. From
    \cref{theorem:construct-irred-over-extension-field}, it takes time $T(dk, p)
    + \tilde{O}(dk^3\log{p})$, where $T(dk, p)$ is time taken for the
    sub-routine which constructs degree $dk$ irreducible polynomial over $\Fp$.
    We use \cref{algo:construct-irred-mod-p} to implement this sub-routine,
    which from \cref{theorem:construct-irred-mod-p} takes time $\tilde{O}(d^4
    k^4 \log^3{p})$. Thus, the overall running time is $\tilde{O}(d^4 k^4
    \log^3{p})$. Let $q = p^k$. Thus, we have given a pseudo-deterministic
    algorithm for constructing irreducible polynomials of degree $d$ over $\Fq$
    in expected time $\tilde{O}(d^4 \log^4{q})$.
\end{proof}

\section{Conclusion}
\label{section:conclusion}

We have shown an efficient pseudo-deterministic algorithm for constructing
irreducible polynomials of degree $d$ over finite field $\Fq$. It is natural to
ask if this algorithm can be derandomized to get a fully deterministic
algorithm. Since our approach heavily relies on fast randomized polynomial
factoring algorithms, and no efficient deterministic factoring algorithms are
known, it is unclear how to derandomize it using the above approach. In fact, we
don't even know how to deterministically construct a quadratic non residue
modulo $p$ ($p$ is prime).

Another interesting question is to compare the hardness of deterministically
factoring polynomials and deterministically constructing irreducible polynomials
over finite fields. As mentioned earlier, Shoup~\cite{shoup-88} had showed that
constructing irreducible polynomials over finite fields can be efficiently (and
deterministically) reduced to factoring polynomials. This suggests that
factoring polynomials is as hard as constructing irreducible polynomials. But
what about the other direction? Would we be able to factor polynomials
efficiently if we could construct irreducible polynomials?

The answer is affirmative in the quadratic case. Suppose we are given a
quadratic non residue $\beta$ modulo $p$. Then we can compute the square roots
of any quadratic residue $\alpha$ module $\Fp$. In other words, given an
irreducible polynomial $X^2 - \beta$, we can factorize $X^2 - \alpha$. This can
be achieved using the Tonelli-Shanks~\cite{shanks-73, tonelli-1891} algorithm
for computing square roots modulo $p$. However, this technique does not easily
generalize to higher degrees $d$, so there isn't enough evidence to confirm that
constructing irreducible polynomials is as hard as factoring polynomials in
general. We believe this is an interesting open question that can shine more
light on the complexity of both these problems.

Gat and Goldwasser~\cite{gat-goldwasser-2011} highlighted the open problem of
pseudo-deterministically constructing $n$-bit prime numbers, which still remains
unsolved. Chen et al.~\cite{chen-lu-oliveira-santhanam-2023} solved this problem
but with the caveat that their algorithm works in the infinitely often regime.
Their algorithm is based on complexity theoretic ideas. In this paper, we gave a
pseudo-deterministic algorithm for constructing irreducible polynomials, which
leverages the structure of irreducible polynomials. Perhaps similarly one could
hope to get an efficient pseudo-deterministic algorithm for constructing primes
using some number theoretic approaches.

\section*{Acknowledgements}

The author would like to thank Mrinal Kumar and Ramprasad Saptharishi for
introducing him to the question of pseudo-deterministic construction of
irreducible polynomials and for the many insightful discussions along the way.  

\bibliography{ref}

\end{document}